\documentclass[a4paper,USenglish,cleveref,autoref,thm-restate]{article}

\usepackage{amsmath}
\usepackage{amsthm}
\usepackage{amssymb} 
\usepackage{fullpage}
\usepackage{thmtools}
\usepackage{thm-restate,amssymb}
\usepackage{hyperref}
\usepackage[T1]{fontenc} 
\usepackage{graphicx} \usepackage{xspace} 
\usepackage[numbers]{natbib} 
\newtheorem{theorem}{Theorem}
\newtheorem{lemma}{Lemma}
\newtheorem{observation}{Observation}
\newtheorem{definition}{Definition}
\newtheorem{corollary}{Corollary}

\title{Sparse Outerstring Graphs Have Logarithmic Treewidth}

\author{Shinwoo An\footnote{Pohang University of Science and Technology, Korea. Email: \href{mailto:shinwooan@postech.ac.kr}{\texttt{shinwooan@postech.ac.kr}}} 
 \hspace{50pt}    Eunjin Oh\footnote{Pohang University of Science and Technology, Korea. Email: \href{mailto:eunjin.oh@postech.ac.kr}{\texttt{{eunjin.oh@postech.ac.kr}}}}
 \hspace{50pt}    Jie Xue\footnote{New York University Shanghai, China. Email: \href{mailto:jiexue@nyu.edu}{\texttt{{jiexue@nyu.edu}}}}
 }

\newcommand{\cin}{\ensuremath{\interior(C)}}
\newcommand{\interior}{\textsf{int}}

\theoremstyle{definition}

\date{}
\begin{document}
\maketitle

\begin{abstract}
    An outerstring graph is the intersection graph of curves lying inside a disk 
    with one endpoint on the boundary of the disk. We show that 
    an outerstring graph with $n$ vertices has treewidth $O(\alpha\log n)$, where $\alpha$ denotes the arboricity of the graph, with an almost matching lower bound of $\Omega(\alpha \log (n/\alpha))$. 
    As a corollary, we show that a $t$-biclique-free outerstring graph has treewidth $O(t(\log t)\log n)$. 
    This leads to polynomial-time algorithms for most of the central NP-complete problems such as \textsc{Independent Set}, \textsc{Vertex Cover}, \textsc{Dominating Set}, \textsc{Feedback Vertex Set}, \textsc{Coloring} for sparse outerstring graphs.  
    Also, we can 
    obtain subexponential-time (exact, parameterized, and approximation) algorithms for various NP-complete problems such as 
    \textsc{Vertex Cover}, \textsc{Feedback Vertex Set} and \textsc{Cycle Packing} for (not necessarily sparse) outerstring graphs. 
\end{abstract}

\section{Introduction}
The \emph{intersection graph} of a family $\mathcal F$ of geometric objects is the graph $G=(V,E)$ such that 
every vertex of $G$ corresponds to an object of $\mathcal F$,
and two vertices of $G$ are connected  by an edge if and only if
their corresponding objects intersect. 
In this case, $\mathcal F$ is called a \emph{geometric representation} of $G$. Notice that a geometric representation of an intersection graph is not necessarily unique. 
There are several popular classes of intersection graphs such as string graphs, unit disk graphs, and disk graphs. 
In the case that $\mathcal F$ is a family of curves in the plane, 
its intersection graph is called a \emph{string graph}. 
Similarly, the intersection graph of a family of unit disks (or disk graphs)
is called a \emph{unit disk graph} (or a \emph{disk graph}). 
Notice that a unit disk graph is a disk graph, and a disk graph is a string graph\footnote{We can represent each disk as a densely spiral-shaped string that covers the interior of the disk.}. 

Geometric intersection graphs have been studied extensively as early as in the 1960s, motivated by the connection with integrated RC circuits~\cite{ehrlich1976intersection, Sinden1966TopologyOT}. 
Most NP-complete problems on general graphs remain NP-complete in geometric intersection graphs (or even in unit disk graphs).
However, recently, it is known that lots of NP-complete problems can be solved in subexponential time on geometric intersection graphs~\cite{an2023faster,an2021feedback,an2024eth,bonnet2019optimality,de2020framework,de2023clique,kisfaludi2022computing,lokshtanov2022subexponential,okrasa2020subexponential}. 
For instance, \textsc{Vertex Cover}, \textsc{3-Coloring} and \textsc{Feedback Vertex Set} can be solved in $2^{O(n^{2/3}\text{polylog }n)}$ time for string graphs with $n$ vertices~\cite{bonnet2019optimality}. 
In the case of disk graphs and unit disk graphs, one can obtain even stronger results. 
There are subexponential-time \emph{parameterized} algorithms for numerous NP-complete problems in this case. 
For instance, one can solve \textsc{Vertex Cover} and \textsc{Feedback Vertex Set} in  $2^{O(\sqrt{k})}n^{O(1)}$ time
for unit disk graphs~\cite{an2021feedback,de2020framework}, and in $2^{O(k^c)}n^{O(1)}$ time for disk graphs~\cite{an2023faster,lokshtanov2022subexponential} for some constant $c<1$,
where $k$ denotes the output size. 

All of the algorithms mentioned above use the fact that 
string graphs have treewidth sublinear in the number of vertices if they do not have a large clique (or biclique).\footnote{The definition of the treewidth is given in Section~\ref{sec:preliminary}.} 
For a string graph $G$ with $n$ vertices that does not contain a biclique of size $t$ has a balanced separator of size $O(\sqrt{t(\log t)n})$~\cite{lee2017separators}, and thus its treewidth is $O(\sqrt{t(\log t)n})$. 
For most central NP-complete problems, there are $2^{O(tw)}n^{O(1)}$-time algorithms for any graphs of treewidth $tw$. 
For subexponential-time algorithms for string graphs given in~\cite{bonnet2019optimality}, 
the authors produce several instances of \emph{sparse} string graphs using the branching technique.
Then they apply $2^{O(tw)}n^{O(1)}$-time algorithms to the instances. 
The subexponential-time parameterized algorithms for unit disk graphs and disk graphs mentioned above
use similar approaches. Using branching, they produce several instances of \emph{sparse} disk graphs. In this case, using the sparsity and the geometric properties of disk graphs, they show that the treewidth of the resulting disk graph depends only on the parameter and the measure for the sparsity of the graph. 

Motivated by these algorithmic applications, an optimal bound on the treewidth of a sparse string graph has been studied extensively for the last two decades~\cite{fox2010separator,lee2017separators,matouvsek2014near,matouvsek2014string}. 
A $t$-biclique-subgraph-free ($K_{t,t}$-free) string graph with $n$ vertices has treewidth $O_t(\sqrt n)$, while a general (not necessarily sparse) string graph with $m$ edges has treewidth $O(\sqrt m)$. 
These bounds are tight for string graphs~\cite{lee2017separators}. One might hope to obtain a better bound for a special subclass of string graphs, which might lead to faster algorithms for central NP-hard problems. However, 
this bound is tight even for sparse unit disk graphs and the intersection graph of axis-parallel segments. 
In particular, a $\sqrt n\times \sqrt n$ grid is a unit disk graph of treewidth $\Theta(\sqrt n)$. 
Also, there is an axis-parallel segment graph of treewidth $\Theta(\sqrt n)$ which does not contain a $K_{2,2}$ as a subgraph.
See Figure~\ref{fig:example}. 

\begin{figure}
	\centering
	\includegraphics[width=0.5\textwidth]{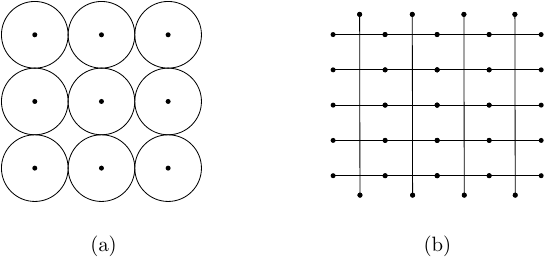}
	\caption{(a) A $\sqrt{n}\times \sqrt{n}$ grid is a unit disk graph of treewidth $\Theta(\sqrt n)$. (b) A sparse axis-parallel segment graph of treewidth $\Theta(\sqrt n)$. It does not contain $K_{2,2}$ as a subgraph. The horizontal segments form $\sqrt n$ rows, and each row consists of $\Theta(\sqrt n)$ horizontal segments. }
		\label{fig:example}
\end{figure}    

In this paper, we focus on outerstring graphs. An \emph{outerstring graph}
is the intersection graph of curves lying inside a disk with one endpoint on the boundary of the disk. 
Outerstring graphs have been studied for about 30 years since they were introduced by Kratochvíl~\cite{kratochvil1991string}. 
There are numerous works on the combinatorial properties of outerstring graphs~\cite{biedl2018size,cardinal2018intersection,fox2012coloring,fox2012string,gavenvciak2018cops,rok2019outerstring},
and efficient algorithms for outerstring graphs~\cite{bhore2022balanced,bose2022computing,keil2017algorithm}. 
Basically, \textsc{Recognition} is NP-hard for outerstring graphs~\cite{kratochvil1991string,middendorf1993weakly}, and all combinatorial properties and algorithms of~\cite{bhore2022balanced} are applicable to outerstring graphs without explicitly provided geometric representation.
On the other hand, the algorithms presented in~\cite{ bose2022computing,keil2017algorithm} run in time polynomial in the complexity of a geometric representation.
Despite these efforts, 
we are not able to find any (exact, parameterized, or approximation) polynomial-time algorithms
for central NP-hard problems specialized to outerstring graphs.

As the concept of the treewidth is a key to obtaining efficient algorithms for geometric intersection graphs, a natural direction for this problem is 
to analyze a tight bound on the treewidth of an outerstring graph. 
Fox and Pach~\cite{fox2012coloring} showed that 
an outerstring graph with $m$ edges has treewidth $O(\min\{\Delta, \sqrt m\})$, where $\Delta$ denotes the maximum degree of the graph.\footnote{Precisely, they showed that an outerstring has a balanced separator of size $O(\min\{\Delta, \sqrt m\})$. By~\cite{dvovrak2019treewidth}, this implies that the treewidth of an outerstring graph is $O(\min\{\Delta, \sqrt m\})$. 
}
Moreover, they showed that this bound is tight for outerstring graphs 
as a split graph containing a clique of size $\Theta(m)$ has treewidth $\Theta(m)$, and the Cayley graph with vertex set $\mathbb{Z}_n$ such that any two vertices of cyclic distance at most $\Delta/2$ are adjacent 
has treewidth $\Theta(\Delta)$. A split graph and such a Cayley graph are all outerstring graphs. Although these examples show that the bound of $O(\min\{\Delta, \sqrt m\})$ is tight, these are dense graphs that contain a clique of size $\Theta(\sqrt m)$. 
Indeed, what we need for algorithmic applications is a bound on the treewidth of sparse outerstring graphs, for instance, $t$-biclique-free outerstring graphs.
There is still room for the improvement of the bound in the sparse regime.


\paragraph{Our results.}
In this paper, we show that 
an outerstring graph $G$ has treewidth $O(\alpha\log n)$, where $\alpha$ is the arboricity of $G$. The arboricity of $G$ is defined as the maximum average degree of the subgraphs of $G$.
Using the previous structural results about outerstring graphs~\cite{fox2012string, lee2017separators}, we show that an outerstring graph which does not contain $K_{t,t}$ as a subgraph has arboricity $O(t\log t)$.  
Thus our main result implies that 
a $t$-biclique-free outerstring graph $G$ has treewidth $O(t(\log t)\log n)$. We emphasize that all of our algorithmic applications mentioned below work even without a geometric representation of an outerstring graph. 

First, we can obtain polynomial-time algorithms for any problem that admits a single-exponential-time algorithm parameterized by treewidth, including \textsc{Independent Set}, 
\textsc{Hamiltonian Cycle},
\textsc{Dominating Set}, and
\textsc{Feedback Vertex Set}, for $t$-biclique-free outerstring graphs for a fixed constant $t$.
All algorithms but the algorithm for \textsc{Independent set} are the first polynomial-time algorithms for these problems in $t$-biclique-free outerstring graphs. Moreover, the algorithm for \textsc{Independent set} is the first polynomial-time \emph{robust} algorithm that does not require the geometric representation of the graph\footnote{Our algorithms for the other problems are also robust. There is an FPT algorithm on \textsc{Independent Set} parameterized by the complexity of the geometric representation of an outerstring graph~\cite{keil2017algorithm}.}.

Our main result can be used for obtaining various algorithms for (possibly dense) outerstring graphs. We design subexponential-time FPT algorithms for \textsc{Vertex Cover} and \textsc{Feedback Vertex Set} on general outerstring graphs work in $2^{O(\sqrt k\log^2 k)}n^{O(1)}$ time, where $k$ is the solution size.
It is known that  \textsc{Vertex Cover} can be solved in time polynomial in 
the complexity of a geometric representation of an outerstring graph~\cite{bhore2022balanced,bose2022computing,keil2017algorithm}. 
But the complexity of a geometric representation of an outerstring graph can be exponential in the number of vertices~\cite{biedl2018size}.
For string graphs, \textsc{Vertex Cover} can be solved in $2^{O(k^{2/3})}n^{O(1)}$ time~\cite{bonnet2019optimality}. On the other hand, \textsc{Feedback Vertex Set} can be solved in $2^{O(n^{2/3})}$ time for string graphs, and no subexponential-time algorithm parameterized by the solution size was known for this problem on string graphs and outerstring graphs prior to our work.

In addition, we can obtain (non-parameterized) $2^{O(\sqrt n\log^2 n)}$-time algorithms for
\textsc{Maximum Induced Matching} and \textsc{List 3-Coloring} on (general) outerstring graphs.
These improve the algorithms for these problems running in $2^{O(n^{2/3})}$ time~\cite{bonnet2019optimality}. (But the algorithms in~\cite{bonnet2019optimality} also work for string graphs.)
Finally, we can design a 4-approximation algorithm for \textsc{Cycle Packing} works in $n^{O(\log\log n)}$ time.
Prior to this work, neither approximation nor parameterized algorithms for 
\textsc{Cycle Packing} on outerstring graphs were known.
For general graphs, \textsc{Cycle Packing} 
 admits an $O(\sqrt{\log n})$-approximation algorithm~\cite{krivelevich2007approximation}, and it is quasi-NP-hard to approximate within a factor of $O(\log^{1/2-\varepsilon} n)$ for any $\varepsilon>0$~\cite{friggstad2011approximability}.

We believe that our result can serve as a starting point for designing efficient algorithms for outerstring graphs. 
Apart from these applications, we also believe that the main result itself is interesting. 
Only a few natural classes of graphs such as sparse $O_k$-free graphs, even-hole-free graphs of bounded degree, 
even-hole-free graphs of bounded clique number, and (theta, triangle)-free graphs are known to have logarithmic treewidth~\cite{aboulker2021tree,abrishami2022induced,bonamy2023sparse,chudnovsky2022induced}.

\paragraph{Related work.}
Outerstring graphs were introduced by~\cite{kratochvil1991string}. 
The class of outerstring graphs is a broad subclass of string graphs, which includes split graphs (graphs whose vertex set can be decomposed into a clique and an independent set), incomparability graphs (graphs representing the incomparability of elements in a partially ordered set), circle graphs (intersection graphs of chords of a circle), 
and ray graphs (intersection graph of rays starting from the $x$-axis).
Moreover, several real-world problems such as PCP routing and railway dispatching can be stated in terms of outerstring graphs~\cite{flier2010vertex,kong2010optimal}.

Lots of NP-hard problems remain NP-hard on outerstring graphs.
For instance, \textsc{Minimum Clique Cover}, \textsc{Coloring}, and 
\textsc{Dominating Set}, and \textsc{Hamiltonian Cycle} 
are NP-hard even on circle graphs~\cite{damaschke1989hamiltonian,garey1980complexity,keil1993complexity}. 
Also, \textsc{Maximum Clique} is NP-hard even on ray graphs~\cite{cabello2013clique}.
\textsc{Recognition} is also NP-hard for outerstring graphs~\cite{kratochvil1991string,middendorf1993weakly}. 
On the other hand, \textsc{Independent Set} can be solved in time polynomial in the complexity of a geometric representation of an outerstring graph if its geometric representation is given~\cite{keil2017algorithm}.
However, there is an outerstring graph that does not admit geometric representation of polynomial complexity~\cite{biedl2018size}. 
It is still unknown if \textsc{Independent Set} can be solved in
time polynomial in the number of vertices in the case that geometric representation is not given. 
All missing proofs and details can be found in the Appendices.

\section{Preliminaries}\label{sec:preliminary}
A \emph{curve} is the \emph{image} of a continuous function from a unit interval to $\mathbb R^2$. 
Under this definition, note that the union of pairwise intersecting curves is also a curve. A simple closed curve $C$ partitions the plane into two disjoint regions.
We call the bounded region the \emph{interior} of $C$ and denote it by $\interior(C)$.
We say a collection $\Gamma$ of curves are \emph{grounded} on a simple closed curve $C$ if every curve is contained in the closure of $\interior(C)$, and 
one of its endpoints is contained in $C$. 
In this case, we say $C$ is a \emph{ground} of the collection of curves.
Also, for a curve of $\Gamma$, one of its endpoints lying on $C$ is called a \emph{ground point} of the curve. 
For an index set $I=\{1,2,3\ldots\}$ and a collection $\Gamma=\{\gamma_i\}_{i\in I}$ of curves grounded on $C$, we say a graph $G$ as an \emph{outerstring graph} if it is identical to the intersection graph of $\Gamma$.
In this case, 
we say $\Gamma$ is a \emph{geometric representation} of $G$.
Throughout this paper, we assume the general position assumption that 
no three curves of $\Gamma$ intersect at a single point, and 
two curves of $\Gamma$ cross at their intersection points. 
These assumptions can be achieved by slightly adjusting the curves without changing the intersection graph. 
See Figure~\ref{fig:normal-curve}.
In the following, to distinguish the curves of $\Gamma$ from other curves not necessarily in $\Gamma$, we call the curves of $\Gamma$ the \emph{strings} of $\Gamma$. 

\begin{figure}
	\centering
	\includegraphics[width=0.70\textwidth]{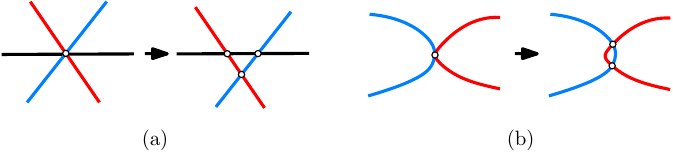}
	\caption{We can always assume the general position assumption.}
		\label{fig:normal-curve}
    \end{figure}    

An induced subgraph of an outerstring graph is also an outerstring graph as an outerstring graph is a geometric intersection graph.
Similarly, an \emph{induced minor} of an outerstring graph is also an outerstring graph although a minor of an outerstring graph is not necessarily an outerstring graph.
An induced minor of a graph $G$ is a graph that can be obtained from $G$ by deleting vertices and contracting edges. 
Whenever we deal with an induced subgraph of $G$, we assume that its underlying geometric representation is a subset of $\Gamma$. 
Similarly, the underlying geometric representation 
of an induced minor of $G$ 
is a set of curves obtained from $\Gamma$ by removing curves
and taking the union of curves. In the case that several curves are merged into a single curve, we choose the ground point of any of them as the ground point of the new curve. 

Let $G$ be the intersection graph of a collection $\Gamma$ of grounded curves. 
For a subset $\Gamma'$ of $\Gamma$, the intersection graph of $\Gamma'$ is an induced subgraph of $G$ by curves in $\Gamma'$. 
We use $G[\Gamma']$ to denote the intersection graph of $\Gamma'$.
Recall that a vertex of $G$ corresponds to a curve of $\Gamma$.
For any subset $U$ of the vertex set of $G$, 
the \emph{geometric representation} of $U$ is defined as
the union of the curves corresponding to the vertices of $U$.

\paragraph{Treewidth and brambles.}
A key notion we use in this paper is the \emph{treewidth} of a graph, which is defined as follows. 
A \emph{tree decomposition} of an undirected graph $G=(V,E)$
is defined as a pair $(T,\beta)$, where $T$ is a tree and $\beta$ is a mapping from nodes of $T$ to subsets of $V$ (we call $\beta(t)$ of $t\in T$ bag) with the following conditions. 

 \begin{itemize}\setlength\itemsep{-0.1em}
		\item  For any vertex $u\in V$, there is at least one bag which contains $u$.
		\item For any edge $(u,v)\in E$, there is at least one bag which contains both $u$ and $v$.
		\item For any vertex $u\in V$, the nodes of $T$ containing $u$ in their bags are connected in $T$.
	\end{itemize}
	
 The $\emph{width}$ of $(T,\beta)$ is defined as the size of its largest bag minus one, and the $\emph{treewidth}$ of $G$ is
	the minimum width of a tree decomposition of $G$.
 Notice that the treewidth of $K_{t,t}$ is $\Theta(t)$, and thus any graph containing $K_{t,t}$ as a subgraph has treewidth $\Omega(t)$. 

Although our main focus is to analyze a tight bound of the treewidth of an outerstring graph, we do not use this definition directly. Instead, we use an alternative characterization of treewidth using the notion of \emph{brambles}.
A \emph{bramble} $\mathcal X$ of a graph $G$ is a family of connected subgraphs of $G$ that all touch each other. Here, we say a subgraph $X$ of $G$ \emph{touches} a subgraph $X'$ of $G$ if there is a common vertex in $V(X)\cap V(X')$, or there is an edge with one endpoint in $V(X)$ and one endpoint in $V(X')$.
A subset $Y$ of $V(G)$ is a \emph{hitting set} of the bramble if $V(X)\cap V(Y)$ contains a common vertex for each subgraph $X$ of $\mathcal X$.
    Then the \emph{order} of bramble is defined as the smallest size of a hitting set of the bramble. 
    Let $\mathcal X$ be a bramble of $G$ of maximum order.
    Then it is known that the order of $\mathcal X$ is exactly the treewidth of $G$ plus one~\cite{cygan2015parameterized}.

\paragraph{Sparse outerstring graphs.} Our main result focuses on \emph{sparse} outerstring graphs. There are several different definitions of the sparsity of a graph. 
As the measure for the sparsity of a graph, we mainly use the \emph{arboricity}, which is 
defined as the minimum number of forests into which its edges can be partitioned. Equivalently, it 
is half of the maximum average degree of the subgraphs of the graph.

Other notions for measuring the sparsity of graphs are \emph{biclique-freeness} and \emph{degeneracy}. 
For an integer $t$, we say a graph is \emph{$t$-biclique-free} if
it does not contain a (not necessarily induced) subgraph isomorphic to $K_{t,t}$. 
We say $t$ as the \emph{size} of the biclique $K_{t,t}$.
We say a graph $G$ is \emph{t-degenerate} if every subgraph of $G$
has a vertex of degree at most $t$. 
In the case of outerstring graphs, all concepts mentioned above are equivalent up to log factors as we will see in Lemma~\ref{lem:sparsity}. 
Moreover, the bounds stated in the following lemma are all tight up to constant factors. 
For the first two statements, as mentioned in~\cite{fox2010separator}, 
the construction of \cite{fox2006bipartite, pach2006comment}
shows that there are $t$-biclique-free
incomparability graphs with $n$ vertices and $\Theta((t\log t)n)$ edges.
Moreover, every vertex in this graph has degree $\Omega(t\log t)$. 
It is known that an incomparability graph is an outerstring graph~\cite{fox2012coloring}, and thus
the first two bounds in the following lemma are tight.
The construction in Section~\ref{sec:lower-bound-crossing-level} gives
a tight lower bound for the last two cases.

\begin{restatable}{lemma}{sparsity}
\label{lem:sparsity}
    For an outerstring graph $G$, the following statements hold. 
    \begin{itemize}\setlength\itemsep{-0.1em}
        \item If $G$ is $t$-biclique-free,
        then $G$ is $O(t\log t)$-degenerate.
        \item If $G$ is $t$-biclique-free, $G$ 
         has arboricity $O(t\log t)$.
         \item If $G$ has arboricity $t$, then $G$ is $2t$-biclique-free.
          \item If $G$ is $t$-degenerate,
        then $G$ is $2t$-biclique-free. 
    \end{itemize}
\end{restatable}
\begin{proof}
For the first two statements, 
    assume that $G$ is $t$-biclique-free. 
    Consider an induced subgraph $H$ of $G$. 
    Note that $H$ is also a $t$-biclique-free outerstring graph. 
    By~\cite{fox2012string} and~\cite{lee2017separators},
    the number of edges of $H$ is $O(t(\log t)|V(H)|)$. Therefore, the average degree of $H$ is $O(t\log t)$, and  
    thus $H$ has a vertex of degree $O(t\log t)$. This implies that $G$ is $O(t\log t)$-degenerate, and it has arboricity $O(t\log t)$.

    Now we consider the contrapositives of the last two statements. 
    If $G$ contains a subgraph $H$ isomorphic to $K_{2t,2t}$, the average degree of $H$ is large than $t$. 
    Therefore, $G$ has arboricity greater than $t$, and thus the third statement holds.
    Moreover, it does not contain a vertex of degree at most $t$, and thus $G$ is not $t$-degenerate.
    Therefore, the lemma holds. 
\end{proof}


\section{Upper Bound on the Treewidth of an Outerstring Graph}\label{sec:upper}
In this section, we show that an outerstring graph $G$ has treewidth $O(\alpha\log n)$, where $\alpha$ denotes the arboricity of $G$. 
Let $\Gamma$ be a geometric representation of $G$, which 
is a collection of curves grounded on a ground $C$.
A key of our proof lies in defining a new notion, called the \emph{crossing-level}, which is a variant of the level in an arrangement. 
The \emph{arrangement} of $\Gamma$ is the subdivision of $\interior(C)$ 
formed by the curves of $\Gamma$ into vertices, edges, and faces. 
Note that the degree of each vertex of the arrangement is at most four by the general position assumption.
For an illustration of the crossing-level, see Figure~\ref{fig:crossing-level}(a). 

\begin{definition}
For a point $p$ in {$\textnormal{\cin}$}, the \emph{crossing-level} of $p$ in $\Gamma$
is the smallest number of different strings that one must cross to reach the ground
from $p$.\footnote{In the case that $p$ lies on a curve of $\Gamma$, 
the starting point $p$ is not considered as a crossing point.}
The \emph{maximum crossing-level} of $\Gamma$ is the maximum of the crossing-levels of all points of $\textnormal{\cin}$. 
\end{definition}

Our proof consists of three steps.
We first show that 
the maximum crossing-level of $\Gamma$ is at most  $3\alpha\log n$.  
In the second step, we show that an outerstring graph $G$ contains a clique minor of size $\Omega(tw)$, where $tw$ denotes the treewidth of $G$.
Then in the third step, using the result of the second step,
we show that 
the maximum crossing-level of $\Gamma$ is at least $\Omega(tw)$.  
By combining the two claims, we conclude that the treewidth of $G$ is $O(\alpha\log n)$.

%

\subsection{Step 1. Upper Bound on the Maximum Crossing-Level} 
\label{sec:upper-bound-crossing-level}
In this subsection, we show that the maximum crossing-level of $\Gamma$ is $O(\alpha\log n)$. 
We use the following observation which immediately comes from the definition of the crossing-level.

\begin{observation} \label{obs:arrange}
Any two points in the same face of the arrangement of $\Gamma$ have the same crossing-level. Also, any two points of the same edge of the arrangement of $\Gamma$ have the same crossing-level unless one of the points is an endpoint of the edge. 
\end{observation}

Let $p$ be a point of $\cin$ that achieves the maximum crossing-level of $\Gamma$, and let $r$ be the crossing-level of $p$.
For each index $0\leq i\leq r$, let $U_i$ be the set of points in $\cin$ of crossing-level at least $i$. Clearly, $U_0$ coincides with $\cin$. 
Notice that $U_i$ contains $p$ for all indices $0\leq i\leq r$. 
Let $R_i$ be the boundary of the connected
component of $U_i$ containing $p$. 
Due to Observation~\ref{obs:arrange}, 
$R_i$ is a Jordan curve consisting of parts of strings of $\Gamma$ for $i>1$, and $R_0$ coincides with the ground. 
See also Figure~\ref{fig:crossing-level}(a).
Then let $\Gamma_i$ be the set of strings of $\Gamma$ intersecting $R_i$. 
Since $R_0,R_1,\ldots R_{r}$ are concentric,
$\Gamma_r\subseteq \Gamma_{r-1}\subseteq \ldots \subseteq \Gamma_0=\Gamma$. 
In the following lemma, we show that the size of $\Gamma_i$ is (almost) geometrically decreasing. Since $|\Gamma_0|=n$, this gives an upper bound of $r$ in terms of input size $n$ and arboricity $\alpha$. 

\begin{figure}
	\centering
	\includegraphics[width=0.55\textwidth]{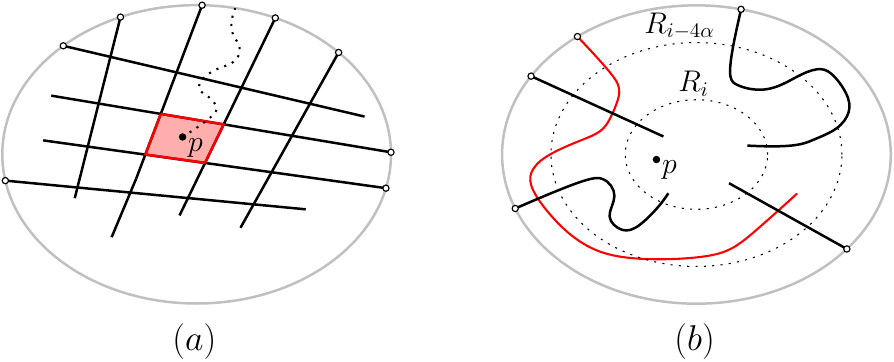}
	\caption{(a) The red region depicts $U_2$. Jordan curve $R_2$ consists of the parts of four strings. (b) 
 All black strings are contained in both $\Gamma_i$ and $\Gamma_{i-4\alpha}$,
 and the red string is contained in $\Gamma_i$ only. }
		\label{fig:crossing-level}
\end{figure}

\begin{lemma} \label{lem:small-crossing-level}
    For each index $i\geq 4\alpha$,  where $\alpha$ is the arboricity of $G$, we have $2|\Gamma_{i}|\leq |\Gamma_{i-4\alpha}|$.
\end{lemma}
\begin{proof} 
We consider the region bounded by $R_i$ and $R_{i-4\alpha}$. See Figure~\ref{fig:crossing-level}(b). 
Note that a string of $\Gamma_i$ intersects both $R_i$ and $R_{i-4\alpha}$ by definition. 
Each string $\gamma$ of $\Gamma_{i}$ is intersected by 
at least $4\alpha$ \emph{different} strings of $\Gamma$ in the region bounded by $R_i$ and $R_{i-4\alpha}$. Otherwise, the first intersection point of $\gamma$ with $R_{i}$ 
from the ground point of $\gamma$ would have the crossing-level less than $i$, which contradicts the fact that this intersection point lies on $R_{i}$. 
Also, observe that a string of $\Gamma$ intersecting $\gamma$ in the region bounded by $R_i$ and $R_{i-4\alpha}$ is contained in $\Gamma_{i-4\alpha}$. 
Let $\mathcal P$ be the set of all pairs $(\gamma,\gamma')$ 
such that a string $\gamma$ of $\Gamma_i$ is intersected by a string $\gamma'
$ of $\Gamma_{i-4\alpha}$. 
Due to the previous observation, we have $|\mathcal P|\geq 4\alpha\cdot |\Gamma_i|$. 

Now we show that $|\mathcal P|\leq 2\alpha \cdot |\Gamma_{i-4\alpha}|$. 
To see this, observe that 
a pair $(\gamma,\gamma')$ of $\mathcal P$ corresponds to an edge of $G[\Gamma_{i-4\alpha}]$. 
This is because $\gamma\in \Gamma_i \subseteq \Gamma_{i-4\alpha}$, and 
$\gamma' \in \Gamma_{i-4\alpha}$. 
Moreover, by construction, an edge of $G[\Gamma_{i-4\alpha}]$ corresponds to
at most two different pairs of  $\mathcal P$. 
Here, notice that two strings of $\Gamma$ may intersect more than once,
but in the construction of $\mathcal P$, 
for all pairs of $\mathcal P$ whose first elements are the same, 
there second elements must be distinct. 
Therefore, an edge of $G[\Gamma_{i-4\alpha}]$ corresponds to at most two pairs of $\mathcal P$. 
Since the arboricity of $G$ is $\alpha$, any subgraph $H$ of $G$ has at most $\alpha\cdot |V(H)|$ edges.
Then the number of edges of $G[\Gamma_{i-4\alpha}]$
is at most $\alpha \cdot |\Gamma_{i-4\alpha}|$. 
Therefore, we have
$|\mathcal P|\leq 2\alpha\cdot |\Gamma_{i-4\alpha}|$, and thus
the lemma holds.
\end{proof}

\begin{lemma} \label{cor:crossing-level}
    The maximum crossing-level of $\Gamma$ is at most $O(\alpha\log n)$. 
\end{lemma}
\begin{proof}
    For the maximum crossing-level $r$ of $\Gamma$, 
    we have $\Gamma_r\neq\emptyset$.  
    By Lemma~\ref{lem:small-crossing-level},
    $2|\Gamma_i|\leq |\Gamma_{i-4\alpha}|$ for all indices $i\geq 4\alpha$. 
    Therefore, $2^{\lfloor r/(4\alpha)\rfloor} \cdot |\Gamma_r|\leq  |\Gamma_0|= n$.
    Thus, $r\leq 4\alpha(\log n+1)$.
\end{proof}

\subsection{Step 2. Existence of a Clique Minor of Size \texorpdfstring{$tw$}{}}
\label{sec:step-2}
In this subsection, we show that an outerstring graph has a clique minor of size $\Omega(tw)$. 
A \emph{clique minor} of $G$ is a clique formed from $G$ by deleting edges, vertices, and contracting edges.
Note that a general ``string" graph may not contain a clique minor of size $\Omega(tw)$. For instance, an $n$ by $n$ rectangular grid graph is a string graph since it is a planar graph and all planar graphs are string graphs~\cite{schaefer2004decidability}. The treewidth of this graph is $\sqrt n$ but it does not have a size-five clique minor.
For any two disjoint sets $A$ and $B$ of vertices of a general graph $H$ with $|A|=|B|$,
an $(A,B)$-\emph{linkage} of $H$ 
is defined as a set of $|A|$ vertex-disjoint paths
connecting every vertex of $A$ and every vertex of $B$. See Figure~\ref{fig:clique-proof}(a). 
We say a set $Q$ of vertices of $H$ is \emph{well-linked}
if for any two disjoint subsets $A$ and $B$ of $Q$ with $|A|=|B|$, 
there is an $(A,B)$-linkage of $H$. 
The following lemma is frequently used in literature without formal proof,
but to make our paper self-contained, we add a short proof 
in Appendix~\ref{apd:sparsity}. 

\begin{restatable}[Folklore]{lemma}{folklore}\label{lem:folklore}
Any graph $H$ of treewidth $tw$ has a well-linked set of size $\Theta(tw)$. 
\end{restatable}
\begin{proof}
    Let $\mathcal X$ be a bramble of $H$ of maximum order. 
    Recall that the order of $\mathcal X$ is $tw+1$ by the characterization of the treewidth mentioned in Section~\ref{sec:preliminary}. 
    Let $Q$ be the smallest hitting set of $\mathcal X$ with $|Q|=tw+1$. 
    We show that $Q$ is well-linked. If this is not the case,
    there is a witness $(A,B)$ for two disjoint sets $A$ and $B$ of $Q$ with $|A|=|B|$.
    That is, the maximum number of vertex-disjoint paths between $A$ and $B$ is less than $|A|$. 
    Among all such witnesses $(A,B)$, we choose the one that minimizes $|A|$. 
    Then no two vertices, one from $A$ and one from $B$, are adjacent. 
    Then by Menger's theorem, there is a vertex set $X$ of size less than $|A|$ such that every path between $A$ and $B$ intersects $X$. 
    
    We show that $X\cup (Q\setminus A)$ or $X\cup (Q\setminus B)$ is a hitting set of $\mathcal X$, which contradicts that $H$ is the smallest hitting set of $\mathcal X$. 
    To show this, observe that for an element $X$ of $\mathcal X$,
    either $X$ intersects $X$, or 
    it is fully contained in a connected component of $G-X$. 
    The elements of $\mathcal X$ of the first type are hit by $X$,
    and thus it is hit by both $X\cup (H\setminus A)$ and $X\cup (H\setminus B)$. 
    The elements of $\mathcal X$ of the second type must be contained in
    the same connected component of $G-X$ as they touch each other.
    Then either $(H\setminus A)$ or $(H\setminus A)$ hits all such elements. 
    Therefore, either $X\cup (H\setminus A)$ or $X\cup (H\setminus B)$ is a hitting set of $\mathcal X$, and thus the lemma holds.     
\end{proof}

Let $X$ be a well-linked set of $G$ of size $\Theta(tw)$. 
Without loss of generality, we may assume that the size of $X$ is a power of four. 
Assume further that they are sorted along $C$ with respect to the ground points of their corresponding curves of $\Gamma$. 
Then we partition $X$ into four equal-sized subsets $X_1, X_2, X_3, X_4$ so that each subset consists of consecutive vertices of $X$ with respect to their corresponding ground points. See Figure~\ref{fig:clique-proof}(b). 
By the well-linkedness of $X$, there are an $(X_1,X_3)$-linkage $\mathcal P_h$ and an $(X_2,X_4)$-linkage $\mathcal P_v$ of $G$. 
Observe that if all vertices in the paths of $\mathcal P_h\cup \mathcal P_v$
are distinct, 
the geometric representation of every path of $\mathcal P_h$ 
crosses the geometric representation of every path of $\mathcal P_v$. 
Then we simply pair each path of $\mathcal P_h$ with a path of $\mathcal P_v$ and then contract all the edges in the two paths to form a single vertex.
In this way, we have $\Omega(tw)$ contracted vertices, which form a clique. 
However, a vertex in a path of $\mathcal P_h$ might appear in a path of $\mathcal P_v$. To handle this case, we choose new sets $\tilde{X}_1, \tilde X_2, \tilde X_3$ and $\tilde X_4$ using $X_1,X_2,X_3,X_4$ as follows.

\begin{figure}
    \centering
    \includegraphics[width=0.9\textwidth]{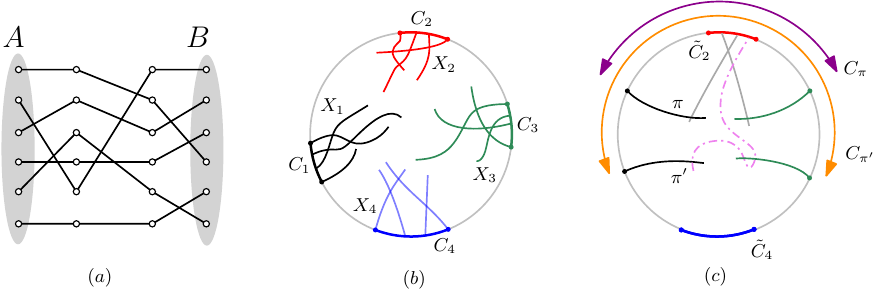}
    \caption{(a) Illustration of an $(A,B)$-linkage. (b) Partition of $X$ into  $X_1,X_2,X_3$ and $X_4$. (c) The case that $C_\pi \subseteq C_{\pi'}$. 
    An internal vertex of $\pi'$ (corresponding to the pink dotted curve) intersects the geometric representation of $\pi$.}
    \label{fig:clique-proof}
\end{figure}

\paragraph{Construction of new sets.}
Recall that the sizes of $X_i$'s are the same. Let $k$ be the size of these sets. 
For $i\in \{1,2,3,4\}$, let $C_i$ be the minimal subarc of $C$ containing
the ground points of all curves corresponding to $X_i$ but not containing the ground point of any curve corresponding to a vertex of $X\setminus X_i$. See Figure~\ref{fig:clique-proof}(b). 
Notice that $C_1, C_2, C_3$ and $C_4$ are pairwise disjoint.
We say two vertex sets $A$ and $B$ of $G$ are \emph{separated} by $C_1$ and $C_3$
if the ground points of the curves corresponding to the vertices of $A$
are contained in one component of $C\setminus \{C_1,C_3\}$, and 
the ground points of the curves corresponding to the vertices of $B$
are contained in the other component of $C\setminus \{C_1,C_3\}$. 
An $(A,B)$-linkage of $G$ is called a \emph{shortest} $(A,B)$-linkage
if its total length is minimum over all $(A,B)$ linkages of $G$. 
Here, the length of a path is defined as the number of vertices in the path. 

Among all pairs $(X_2', X_4')$ of disjoint vertex sets of $V(G)\setminus (X_1\cup X_3)$ separated by $C_1$ and $C_3$
with $|X_2'|=|X_4'|=k$, we choose the one that minimizes the total length of the shortest $(X_2',X_4')$-linkage.
Notice that such a pair always exists since $(X_2,X_4)$ satisfies the conditions for being a candidate. 
Let $\tilde{X}_2$ and $\tilde X_4$ be the resulting sets. Then let $\tilde C_2$ (and $\tilde C_4$) be
the minimal subarc of $C$ containing the ground points of all curves corresponding to $\tilde X_2$ (and $\tilde X_4$) and not containing any other curve of $\tilde X_4$ (and $\tilde X_2$). 
 Then
among all pairs $(X_1',X_3')$ of disjoint vertex sets of $V(G)\setminus (\tilde X_2\cup \tilde X_4)$ separated by $\tilde C_2$ and $\tilde C_4$
with $|X_1'|=|X_3'|=k$,
we choose the one that minimizes the total length of the shortest $(X_1',X_3')$-linkage. 
Notice that such a pair always exists since $(X_1,X_3)$ satisfies the conditions for being a candidate. 
Let $\tilde{X}_1$ and $\tilde X_3$ be the resulting sets. 
By construction, all resulting sets are pairwise disjoint. 

Then a shortest $(\tilde X_1, \tilde X_3)$-linkage satisfies the following property. By changing the roles of $\tilde C_2\cup \tilde C_4$ and $C_1\cup C_3$,
we can show that the following also holds for a shortest $(\tilde X_2, \tilde X_4)$-linkage. 

\begin{lemma}\label{lem:short-linkage}
    Let $\mathcal P$ be a shortest $(\tilde X_1, \tilde X_3)$-linkage. 
    Then the set of paths of $\mathcal P$ of length at least three
    can be decomposed into two subsets such that 
    for any two paths of $\mathcal P$ in the same subset,
    their geometric representations intersect. 
\end{lemma}
\begin{proof}
    Let $\pi$ be a path of $\mathcal P$ of length at least three. 
    For any internal vertex of $\pi$,
    its corresponding curve of $\Gamma$ has the ground point on $\tilde C_2\cup \tilde C_4$.
    Otherwise, we partition $\pi$ into two subpaths $\pi_1$ and $\pi_2$ whose common endpoint $v$ lies outside of $\tilde C_2\cup \tilde C_4$.
    Assume, without loss of generality, it lies the subarc of $C\setminus (\tilde C_2\cup \tilde C_4)$ containing $C_1$. In this case, the other endpoint of $\pi_i$($i$=1 or 2) lies in the other subarc of $C\setminus (\tilde C_2\cup \tilde C_4)$. Then we simply replace $\pi$ with $\pi_i$. In this way, we can decrease the total length of paths of $\mathcal P$ without any conditions for $\mathcal P$.
    This contradicts the choice of $\mathcal P$.
    Then we decompose the paths of $\mathcal P$ of length at least three
    into two subsets $\mathcal P'$ and $\mathcal P''$ such that 
    $\mathcal P'$ consists of all paths of $\mathcal P$ of length at least three containing at least one internal vertex whose corresponding curve has its ground point on $\tilde C_2$, and $\mathcal P''$ consists of all paths of $\mathcal P$ of length at least three not contained in $\mathcal P'$. 

    Now we show that for any two paths $\pi$ and $\pi'$ of $\mathcal P'$,
    their geometric representations intersect. The other case can be handled symmetrically by changing the roles of $\tilde C_2$ and $\tilde C_4$.  
    Let $C_\pi$ be the subarc of $C$ containing $\tilde C_2$ lying between the
    ground points of the strings corresponding to the endpoints of $\pi$. See Figure~\ref{fig:clique-proof}(c). 
    Similarly, let $C_{\pi'}$ be the subarc of $C$ containing $\tilde C_2$ lying between the ground point of the strings corresponding to the endpoints of $\pi$.  
    If neither $C_\pi$ nor $C_{\pi'}$ contains the other,
    $\pi$ and $\pi'$ cross because the geometric representation of $\pi$ (and $\pi'$) is connected. 
    If this is not the case, without loss of generality, assume that $C_\pi$ is contained in $C_{\pi'}$. Since $\pi'$ has an internal vertex whose corresponding string has its ground point on $\tilde C_2$, the geometric representation of $\pi$ must intersect the string. See the pink dotted curve in Figure~\ref{fig:clique-proof}(c). Therefore,
    the lemma holds. 
\end{proof}

\begin{figure}
    \centering
    \includegraphics[width=0.65\textwidth]{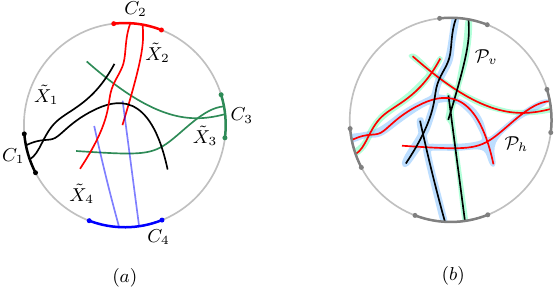}
    \caption{(a) Illustration of $\tilde X_1, \tilde X_2, \tilde X_3$ and $\tilde X_4$. (b) Illustration of length-two paths of $\mathcal P_v$ and $\mathcal P_h$. Two geometric representations of two pairs of paths of $\mathcal P_h$ and $\mathcal P_v$ intersect.}
    \label{fig:clique-minor}
\end{figure}

Now we are ready to show that $G$ has a clique minor of size $\Omega(tw)$.
Recall that $\tilde X_i$'s have size $k=\Omega(tw)$. 
Let $\mathcal P_v$ be a shortest $(\tilde X_1, \tilde X_3)$-linkage,
and $\mathcal P_h$ be a shortest $(\tilde X_2, \tilde X_4)$-linkage.
If $\mathcal P_v$ (or $\mathcal P_h$) contains $k/2$ paths of length at least three, then we are done.
To see this, observe that $\mathcal P_v$ (or $\mathcal P_h$) 
contains at least $k/4$ 
 paths whose geometric representations intersect by Lemma~\ref{lem:short-linkage}. 
For each such path, we contract all edges in the path into a single vertex.
Then at least $k/4$ vertices form a clique,
and thus $G$ has a clique minor of size $k/4$.

Thus we assume each of $\mathcal P_v$ and $\mathcal P_h$ contains
at least $k/2$ paths of length \emph{at most} two. 
Recall that $\tilde X_i\cap \tilde X_j=\emptyset$ for any two distinct indices $i,j\in\{1,2,3,4\}$.
This implies that every such path of $\mathcal P_v$ and $\mathcal P_h$ has length \emph{exactly} two. 
Moreover, the paths of $\mathcal P_v\cup \mathcal P_h$ of length exactly two
are pairwise vertex-disjoint since $\tilde X_i\cap \tilde X_j=\emptyset$.
Since the geometric representation of every path of $\mathcal P_h$ 
crosses the geometric representation of every path of $\mathcal P_v$, 
we pair each path of $\mathcal P_h$ with a path of $\mathcal P_v$ and then contract all the edges in the two paths to form a single vertex. 
In this way, we have $k/2$ contracted vertices, which form a clique. See Figure~\ref{fig:clique-minor}.
Therefore, in any case, $G$ has a clique minor of size $k/2 =\Omega(tw)$. 


\subsection{Step 3. Lower Bound on the Maximum Crossing-Level}
As the third step, we show that the maximum crossing-level of $\Gamma$ is $\Omega(tw)$. 
First, we reduce the general case to the case that $\Gamma$ consists of 
\emph{double-grounded circularly ordered curves},
and then we analyze the maximum crossing-level of $\Gamma$ in this case.
A \emph{double-grounded} curve is a curve having both endpoints on $C$.
We say that double-grounded curves $\gamma_1,\gamma_2,\ldots,\gamma_k$ are \emph{circularly ordered}
if $x_1,x_2,\ldots,x_k,y_1,y_2,\ldots,y_{k-1}$ and $y_k$ lie on $C$ in the counterclockwise order,
where $x_i$ and $y_i$ denote two endpoints  of $\gamma_i$ for $i=1,\ldots,k$.
See Figure~\ref{fig:largeclique}(a).

\begin{figure}
	\centering
	\includegraphics[width=0.8\textwidth]{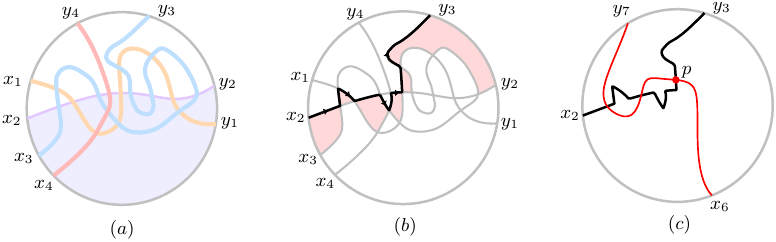}
	\caption{(a) Illustration of double-grounded circularly ordered curves. 
 The colored region is $h_2$. 
 (b-c) Illustrates $P_1$ (black curve) and $P_2$ (red curve). The crossing-level of $p$ is at least $k/2$. }
		\label{fig:largeclique}
\end{figure}  

\paragraph{Reduction to the double-grounded circularly ordered case.}
Let $k$ be the largest integer such that $G$ has a clique minor of size $4k$.
Notice that $k=\Omega(tw)$ by Section~\ref{sec:step-2}. 
Then $G$ has a \emph{model} $\mu$ of $K_{4k}$, that is, 
there is a function $\mu$ that maps the vertices of $K_{4k}$ to 
vertex-disjoint connected subgraphs of $G$
such that, for any two vertices $u$ and $v$ of $K_k$, $\mu(u)$ and $\mu(v)$ are adjacent in $G$. In other words, the edges of $\mu(v)$ are contracted into $v$ in the construction of the clique minor of $G$ of size $k$. 

For a vertex $v$ of $K_{4k}$, recall that the geometric representation of $v$ is defined
as the union of the strings of $\Gamma$ corresponding to the vertices of $\mu(v)$. 
We choose an arbitrary one as its ground point of $v$ (and its geometric representation).
Let $\langle v_1,v_2,\ldots, v_{4k}\rangle$ be the sequence of vertices of $K_{4k}$
sorted with respect to their ground points along $C$ in the counterclockwise order. 
We pair a vertex $v_i$ in the first half of the sequence 
with $v_{i+2k}$, and then find a curve
in the union of their geometric representations. 
More specifically, 
for each index $i$ with $i=1,\ldots,2k$, consider the union of the geometric representations of $v_i$ and $v_{i+2k}$. Notice that the union consists of a single connected component since $v_i$ and $v_{i+2k}$ are adjacent in $K_{4k}$. 
Let $x_i$ be the ground point of $v_i$, and let $y_i$ be the ground point of $v_{i+2k}$. Then the union contains a curve with endpoints $x_i$ and $y_i$.
Let $\gamma_i$ be a curve between $x_i$ and $y_i$ contained in the union. 
By construction, $\gamma_1,\gamma_2,\ldots,\gamma_{2k}$ are doubly-grounded circularly ordered curves. 
Due to the following lemma, it is sufficient to show that
the maximum crossing-level of $k$ double-grounded circularly ordered curves
is $\Omega(k)$. 
\begin{restatable}{lemma}{dgcoc}\label{lem:dgcoc}
    The max crossing-level of $\Gamma$ is 
    at least the max crossing-level of $\{\gamma_1,\ldots,\gamma_{2k}\}$. 
\end{restatable}
\begin{proof}
    Let $q$ be a point in $\cin$, and let $\ell(q)$ be the crossing-level of $q$ in 
    $\{\gamma_1,\gamma_2,\ldots,\gamma_{2k}\}$. 
    We show that the crossing-level of $q$ in $\Gamma$ is at least $\ell(q)$,
    which implies the lemma. 
    Consider a curve $\pi$ connecting $q$ and $C$.
    The number of curves of $\{\gamma_1,\gamma_2,\ldots,\gamma_{2k}\}$
    intersected by $\pi$ is at least $\ell(q)$. 
    By the construction of $\{\gamma_1,\gamma_2,\ldots,\gamma_{2k}\}$,
    for any point in $\gamma_i$ and any point in $\gamma_j$ with $i\neq j$,
    they come from different strings of $\Gamma$. 
    Therefore, the number of strings of $\Gamma$ intersected by $\pi$
    is at least $\ell(q)$, and thus the lemma holds. 
\end{proof}

\paragraph{Analysis of the double-grounded circularly ordered case.}
In the following, we focus on the set $\{\gamma_1,\ldots,\gamma_{2k}\}$ of double-grounded circularly ordered curves along $C$. 
To make the description easier, we assume that $k$ is even. 
   We give a direction to each path $\gamma_i$ from $x_i$ to $y_i$. 
    Then every edge of the arrangement of $\{\gamma_1,\ldots,\gamma_{2k}\}$
    has its direction.
    By the general position assumption, every vertex of the arrangement has
    two incoming arcs and two outgoing arcs. 
    Then $\gamma_i$ subdivides $\cin$ into two regions. The (closed) region 
    lying locally to the right of $\gamma_i$ is denoted by $h_i$.  
    See Figure~\ref{fig:largeclique}(a). 
    Let $H_1=\{h_1,h_2,\ldots h_{k}\}$ and $H_2=\{h_{k+1},h_{k+2},\ldots, h_{2k}\}$. 
    Let $P_i$ be the set of points contained in exactly $k/2$ regions of $H_i$ 
    and lying on the boundary curves of the regions of $H_i$ for $i=1,2$. 
    Then $P_i$ forms a simple curve, but in the following, we prove the weaker property that $P_i$ \emph{contains} a curve as it is sufficient for our purpose. 
    
\begin{restatable}{lemma}{simple}\label{lem:simple}
    The set $P_i$ contains a simple curve connecting two points on $C$ for $i=1,2$. 
\end{restatable}
\begin{proof}
        We prove this for $P_1$ only. The other case can be handled similarly. 
    We traverse the arrangement of the curves of $\{\gamma_1,\ldots,\gamma_{k}\}$ along a simple curve as follows.
     Starting from $x_{k/2}$, we move along $\gamma_{k/2}$ until we reach a vertex of the arrangement. Whenever we meet a vertex of the arrangement,
     this vertex is an intersection point between the curve we are traversing and another curve in $\{\gamma_1,\ldots,\gamma_{k}\}$.
     Then we follow the outgoing arc of the
     vertex lying on the new curve. We repeat this until we reach a point of $C$. 
     Let $\rho$ be the set of points we have visited in this way. 
     We claim that $\rho$ is contained in $P_i$, and it is a simple curve,
     which implies the lemma. 

     We first show that each (directed) arc of $\rho$ is contained in $P_i$. 
     As $x_{k/2}$ is contained in exactly $k/2$ regions of $H_1$, the claim holds for the first arc. 
     Consider two consecutive arcs $\eta$ and $\eta'$ of $\rho$.
     Assume that the claim holds for $\eta$, and we show that the claim also holds for $\eta'$. 
     Let $h$ and $h'$ be the two regions of $H_1$ such that
     the boundary curve of $h$ contains $\eta$, and the boundary curve of $h'$ contains $\eta'$. 
     The regions of $H_1\setminus\{h,h'\}$ containing $\eta$ also contains $\eta'$, and vice versa.
     On the other hand, 
     there are two possibilities: (1) $h$ contains both $\eta$ and $\eta'$, and $h'$ contains both $\eta$ and $\eta'$, or (2) $h$ contains $\eta$ only, and $h'$ contains $\eta'$ only. 
     Therefore, the number of regions of $H_1$ containing $\eta'$ (and $\eta$) 
     are the same,
     and thus $\rho$ is contained in $P_i$. See also Figure~\ref{fig:largeclique}(b).

    Moreover, $\rho$ is a simple curve, that is, it does not contain a cycle. 
    If this is not the case, $\rho$ contains three arcs of the arrangement sharing a common vertex.   
    Two of them are contained in the boundary curve of the
    same region $h$ of $H_1$. Let $h'$ be the region of $H_1$ whose boundary curve
    contains the other arc. 
    The regions of $H_1\setminus \{h'\}$ containing one arc coming from $h$
    contain the other arc coming from $h$, and vice versa. On the other hand, 
    one arc coming from $h$ is not contained in $h'$, and 
    the other arc coming from $h$ is contained in $h'$.
    Therefore, not both arcs coming from $h$ are contained in $P_1$,
    which contradicts that $\rho$ contains both arcs.
    Therefore, $\rho$ does not contain a cycle, and thus it ends at a point on $C$.
\end{proof}

\begin{restatable}{lemma}{level}\label{lem:level}
    We have $P_1\cap P_2\neq\emptyset$, and a point in $P_1\cap P_2$ has a crossing-level at least $k/2$.
\end{restatable}
\begin{proof}
     Notice that $P_1\cap C$ consists of two points $x_{k/2}$ and $y_{k/2+1}$.
    Therefore, $P_1$ contains a simple curve  connecting $x_{k/2}$ and $y_{k/2+1}$.
    Similarly, $P_2$ contains a simple curve connecting $x_{3k/2}$ and $y_{3k/2+1}$. 
    Since $x_{k/2}, x_{3k/2}, y_{k/2+1}$ and $y_{3k/2+1}$ lie on $C$ in the counterclockwise order, the two simple curves cross. 
    See Figure~\ref{fig:largeclique}(c). 
    Notice that an intersection point $p$ 
    of the two simple curves is contained in exactly $k/2$ regions of $H_i$ for $i=1,2$.

    Now we show that the crossing-level of $p$ is at least $k/2$. 
    Let $z$ be a point in $C$. It suffices to show that 
    any curve connecting $p$ and $z$ intersects at least $k/2$ curves of $\{\gamma_1,\ldots,\gamma_{2k}\}$.
    For two points $x'$ and $y'$ on $C$, we let $C[x',y']$ be the circular arc of $C$ 
    from $x'$ to $y'$ in the counterclockwise direction. 
    We show this through case studies of the position of $z$.
    Consider the case that $z\in C[x_{2k},y_1]$. Recall that 
    the common intersection of all regions of $H_1\cup H_2$ contains $C[x_{2k},y_1]$. 
    Recall also that the number of regions of $H$ not containing $p$ is exactly $k$.
    The region $h_i$ of $H$ not containing $p$ is defined by 
    the double-grounded curve $\gamma_i$. Then $\gamma_i$ separates $p$ and $z$. In this case, any curve connecting $p$ and $z$
    intersects at least $k$ different curves of $\{\gamma_1,\gamma_2,\ldots,\gamma_{2k}\}$, and thus we are done. 

    Now consider the case that $z\in C[x_1,x_{k}]\cup C[y_1,y_k]$. 
    Then the regions of $H_2$ containing $p$ do not contain $z$. 
    Therefore,
    at least $k/2$ curves of $\{\gamma_{k+1},\gamma_2,\ldots,\gamma_{2k}\}$ separate $p$ and $z$, and thus any curve connecting $p$ and $z$ intersects at least $k/2$
    different curves of $\{\gamma_1,\gamma_2,\ldots,\gamma_{2k}\}$.

    The other cases can be handled similarly.
    Specifically, the case that $z\in C[y_{2k},x_1]$ can be handled symmetrically to the first case by changing the roles of the regions of $H$ not containing $p$ and the regions of $H$ containing $p$. 
    The case that  $z\in C[x_k,x_{2k}]\cup C[y_k,y_{2k}]$ can be handled symmetrically
    to the second case by changing the roles of $H_1$ and $H_2$.   
\end{proof}

In this way, we can show that for an outerstring graph $G$ with a geometric representation $\Gamma$,
    the maximum crossing-level in $\Gamma$ is $\Omega(tw)$, where $tw$ denotes the treewidth of $G$. 

\begin{theorem}\label{thm:arboricity}
    The treewidth of an outerstring graph is $O(\alpha\log n)$, where $\alpha$ and $n$ denote the arboricity and the number of vertices of the graph, respectively. 
\end{theorem}
\begin{corollary}\label{cor:biclique-free}
 The treewidth of a $K_{t,t}$-free outerstring graph is $O(t(\log t)\log n)$. 
\end{corollary}

\section{Lower Bound of the Treewidth of an Outerstring Graph}
\label{sec:lower-bound-crossing-level}
In this section, we show that the bound of Theorem~\ref{thm:arboricity} is almost tight by constructing an \emph{outersegment graph} $G$ with treewidth $\Theta(\alpha\log (n/\alpha))$, where $\alpha$ is the arboricity of $G$. 
An outersegment graph $G$ is the intersection graph of \emph{segments} grounded on $C$.
We say an outersegment graph is $k$-directional if it has a geometric representation
consisting of grounded segments of $k$ orientations. 
To make the description easier, we first present a 2-directional outersegment graph with arboricity $O(1)$ of treewidth $\Omega(\log n)$. 

\begin{restatable}{lemma}{lowerbound}\label{lem:lowerbound-constant}
    There is a 2-directional outersegment graph with arboricity $O(1)$ and treewidth 
    $\Omega(\log n)$, where $n$ denotes the number of its vertices.
\end{restatable}
\begin{proof}
    Let $n=2^m$. 
    We consider the $x$-axis as the ground, and the region lying above the $x$-axis 
    as the interior of the ground. 
    A \emph{V-shape curve} is the union of two equal-length outersegments 
    sharing a common ground point, with the other endpoints having the same $y$-coordinates. Note that the angle of the two outersegments is $\pi/3$. 
    The width of a V-shaped curve is defined as the length of each outersegment.   
    A \emph{folk} is the union of equal-width V-shaped curves 
    which form a simple curve. 
    See Figure~\ref{fig:lowerbound}(a). 
    The \emph{size} of folk is the number of V-shaped curves in the fork, and the \emph{width} 
    of a folk is the size of each V-shaped curve in the folk. 
    For an index $i$ with $1\leq i \leq m$, 
    let $F_i$ be the folk of width $2^{i}$ and size $n/2^{i}$ whose leftmost point lies on the $y$-axis. 
    Then we set $\Gamma$ as the set of all outersegments of the folks $F_1,F_2,\ldots,F_m$. See Figure~\ref{fig:lowerbound}(b).
    Notice that
    the ground points of 
    the outersegments of $F_i$ have $x$-coordinates $2^{i-1}+2^{i}\cdot j$ for  integers $j$ with $0\leq j<n/2^i-1$. 
    Therefore, no two segments from different folks share their ground points. 

    Its intersection graph $G$ is clearly a 2-directional outersegment graph.
    The number of vertices of $G$ is $\Theta(n)$
    since the number of segments of $F_i$ is $2\cdot (n/2^{i})$. 
    The arboricity of $G$ is $O(1)$ since each segment of $F_i$ intersects
    at most one segment of $F_{i+1}\cup F_{i+2}\cup\ldots \cup F_m$ for any index $i$ with $1\leq i\leq m$. Then the number of edges of any subgraph of $G$ is linear in the number of its vertices, and thus the arboricity of $G$ is $O(1)$.
    Consider the minor of $G$ obtained by contracting all
    edges coming from $F_i$ for each index $i$ with $1\leq i\leq m$.
    Since 
    the folks are pairwise intersecting, the resulting minor is a clique of size $m=\log n$. 
    Since 
    the treewidth of $G$ is at least the treewidth of any of its minors,
    the treewidth of $G$ is $\Omega(\log n)$.
 \end{proof}

\begin{figure}
	\centering
	\includegraphics[width=0.80\textwidth]{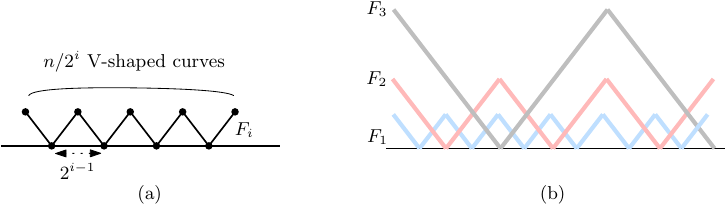}
	\caption{\small (a) A folk of width $2^{i-1}$ and size $n/2^i$.
 (b) Illustration of $F_1, F_2$ and $F_3$. }
		\label{fig:lowerbound}
\end{figure}    

Notice that this example is $2$-biclique-free. 
We can generalize this example to an outersegment graph with an arbitrary arboricity $\alpha$ whose treewidth is $\Theta(\alpha\log (n/\alpha))$ by simply copying all segments $\alpha$ times.
Notice that we can slightly perturb the copied segments without changing their intersection graphs so that any two segments intersect at most once. 

\begin{restatable}{lemma}{lowerboundgeneral}\label{lem:tw-lowerbound}
    For any integer $\alpha$, there is an outersegment graph with arboricity $\alpha$ and treewidth 
    $\Omega(\alpha\log (n/\alpha))$, where $n$ denotes the number of its vertices.
\end{restatable}
\begin{proof}
Let $\Gamma$ be the set of outersegments we constructed in Lemma~\ref{lem:lowerbound-constant}. 
We copy each outersegment of $\Gamma$ into $\alpha$ copies. Then these $\alpha$ curves are pairwise intersecting.
Then the intersection graph of the resulting outersegments has arboricity $\alpha$,
and it has $\alpha n$ vertices. 
Moreover, the resulting set $\Gamma$ can be decomposed into $\alpha \log n$ pairwise intersecting folks, and they form a clique minor of size $\alpha\log n$.
Therefore, the treewidth of the resulting graph is $\alpha \log (N/\alpha)$, where $N$ denotes the number of vertices of the resulting graph. 
\end{proof}

\section{Algorithmic Applications} \label{sec:application}
In this section, we present algorithmic applications of Corollary~\ref{cor:biclique-free}.
Immediate consequences of  Corollary~\ref{cor:biclique-free} are as follows.
We compute a tree decomposition of width $2\cdot tw$ in $2^{O(tw)}n$ time~\cite{korhonen2023single}. 
Then 
lots of central NP-hard problems can be solved in $2^{O(tw)}n^{O(1)}$ time~\cite{cygan2015parameterized} including the problems mentioned in the following corollary (except for \textsc{Coloring}). 
The definitions of the problems are as follows. 
We are given a graph $G$. Then 
\begin{itemize}\setlength\itemsep{-0.1em}
\item \textsc{Independent Set} asks for a maximum-sized set of pairwise non-adjacent vertices.
\item \textsc{Hamiltonian Cycle} asks for a cycle that visits all the vertices of the graph.
\item \textsc{Vertex Cover} asks for a minimum-sized set of vertices that cover all the edges. 
\item \textsc{Dominating Set} asks for a minimum-sized set $D$ 
of vertices such that every other vertex is adjacent to $D$. 
\item \textsc{Feedback Vertex Set} asks for a minimum-sized vertex set
whose removal makes $G$ acyclic. 
\item \textsc{Coloring} asks for a coloring of the vertices of $G$ with the smallest number of colors. 
\item \textsc{Maximum Induced Matching} asks for a maximum-sized edge set $S$ such that no two such edges are joined by an edge of $G$. 
\item \textsc{List 3-Coloring} asks for a function $\bar f: V(G)\rightarrow \{1,2,3\}$
such that
$\bar f(v)\in f(v)$ for all vertices $v$ of $V(G)$, and $\bar f(u)\neq \bar f(v)$ for
all edges $(u,v)$ of $G$. An input of this problem consists of a graph $G$ and a function $f: V(G)\rightarrow 2^{\{1,2,3\}}$. 
\item \textsc{Cycle Packing} asks for a maximum-sized set of vertex-disjoint cycles of $G$. 
\item \textsc{Minimum Clique Cover} asks for a partition of $V(G)$ into a smallest number of cliques. 
\item \textsc{Maximum Clique} asks for a set of pairwise adjacent vertices.
\item \textsc{Recognition} asks for checking if $G$ is an outerstring graph.
\end{itemize}

\begin{restatable}{corollary}{poly}\label{cor:poly}
\textsc{Independent Set}, 
\textsc{Hamiltonian Cycle},
\textsc{Dominating Set}, 
\textsc{Feedback Vertex Set}
and
\textsc{Coloring}
can be solved in polynomial time on $t$-biclique-free outerstring graphs for a constant $t$.
\end{restatable}
\begin{proof}
All problems mentioned above, except for \textsc{Coloring}, admit polynomial-time algorithms
for graphs with logarithmic treewidth. More specifically, they are based on dynamic programming over a tree decomposition of width $tw=O(\log n)$, and they run in $2^{O(tw)}n^{O(1)}=n^{O(1)}$ time in our case.
For details, see~\cite[Chapter 7.3 and Chapter 11.2]{cygan2015parameterized}. 

In the case of \textsc{Coloring}, recall that a $t$-biclique-free outerstring graph is $O(t\log t)$-degenerate by Lemma~\ref{lem:sparsity}. Thus it is $O(t\log t)$-colorable. To see this, consider the following constructive argument. In each iteration, we find a minimum degree vertex $v$ and compute a minimum coloring of $G-\{v\}$ recursively. Then since the degree of $v$ is $O(t\log t)$, we can color $v$
using a color not used by any neighbor of $v$ once we have $\Omega(t\log t)$ colors. 
This implies that $G$ is $O(t\log t)$-colorable.
For a fixed constant $q$, we can compute a coloring of a graph with treewidth $tw$ with $q$ colors
in $2^{O(tw)}n^{O(1)}$ time~\cite[Chapter 7.3]{cygan2015parameterized}. Therefore,
we can solve \textsc{Coloring} in polynomial time in the case that $t$ is a constant. 
\end{proof}


\paragraph{Subexponential-time FPT algorithms for (not necessarily sparse) outerstring graphs.}
We can obtain subexponential-time FPT algorithms for \textsc{Vertex Cover} and \textsc{Feedback Vertex Set} on general outerstring graphs.
To do this, we first compute the polynomial kernel with respect to the solution size $k$. Then we compute biclique of size $t=\sqrt k$ in a brute-force manner, and then apply branching with respect to the partial solution in the biclique. 

\begin{restatable}{corollary}{vc}\label{lem:vc}
\textsc{Vertex Cover} parameterized by the solution size $k$ can be solved in $2^{O(\sqrt k\log^2 k)}n^{O(1)}$ time for (not necessarily sparse) outerstring graphs. 
\end{restatable}
\begin{proof}
    Let $G$ be an outerstring graph with $n$ vertices.
    \textsc{Vertex Cover} admits a kernel of size $k$~\cite{cygan2015parameterized},
    that is, there is an instance $(H,k')$ consisting of an induced minor $H$ of $G$
    with $O(k)$ vertices
    and an integer $k'\leq k$ such that
    $H$ has a vertex cover of size $k'$ if and only if $G$ has 
    a vertex cover of size $k$.
    Thus it is sufficient to solve \textsc{Vertex Cover} for $(H,k')$. Notice that $H$ is also an outerstring graph. 

    Let $t=\sqrt{k}$. Given an instance $(H,k')$, we first compute a biclique of size $t$ in $H$ in a brute-force fashion if it exists. This takes $k^t=2^{O(\sqrt k\log k)}$ time. Let $A,B$ be the sets of $V(H)$ which form a biclique of size $t$. Notice that a vertex cover of $H$ must contain either $A$ or $B$. We branch on whether $A$ is contained in an optimal vertex cover, or $B$ is contained in an optimal vertex cover.
    For this, we produce two instances $(H-A, k'-t)$ and $(H-B, k'-t)$.
    If $H$ does not contain a biclique of size $t$, we apply a polynomial-time algorithm for computing a minimum vertex cover for a $t$-biclique-free outerstring graph. 

    The number of instances we produce is $2^{O(k/t)}$, which is $2^{O(\sqrt k)}$.
    For each instance, we compute a biclique of size $t$ in $2^{O(\sqrt k\log k)}$ time. If it does not exist, we compute a minimum vertex cover in $O(2^{t(\log t)\log k})=2^{O(\sqrt k\log^2 k)}$ time.
\end{proof}

\begin{restatable}{corollary}{fvs}\label{lem:fvs}
\textsc{Feedback Vertex Set} parameterized by the solution size $k$ can be solved in $2^{O(\sqrt k\log^2 k)}n^{O(1)}$ time for (not necessarily sparse) outerstring graphs. 
\end{restatable}
\begin{proof}
     Let $G$ be an outerstring graph with $n$ vertices.
    \textsc{Feedback Vertex Set} admits a kernel of size $k^2$~\cite{cygan2015parameterized}.
    More specifically, there is an instance $(H,k')$ consisting of an induced minor $H$ of $G$ 
    with $O(k^2)$ vertices
    and an integer $k'\leq k$ such that
    $H$ has a feedback vertex set of size $k'$ if and only if $G$ has 
    a feedback vertex set of size $k$.
    Thus it is sufficient to solve \textsc{Feedback Vertex Set} for $(H,k')$. Notice that $H$ is also an outerstring graph. 

    Let $t=\sqrt{k}$. Given an instance $(H,k')$, we first compute a biclique of size $t$ in $H$ in a brute-force fashion if it exists. This takes $|V(H)|^{2t}=2^{O(\sqrt k\log k)}$ time. Let $A,B$ be the sets of $V(H)$ which form a biclique of size $t$. Notice that a feedback vertex set of $H$ must contain either all but one vertex of $A$ or all but one vertex of $B$. In this case, we produce $2t$ instances for branching on which 
    vertices of $A\cup B$ are contained in an optimal feedback vertex set. 
    For this, we produce 
    $t$ instances $(H\setminus(A\setminus\{v\}), k'-(t-1))$ for all vertices $v\in A$ and
    $t$ instances $(H\setminus(B\setminus\{v'\}), k'-(t-1))$ for all vertices $v'\in B$. 
    If $H$ does not contain a biclique of size $t$, we apply a polynomial-time algorithm for computing a minimum feedback vertex set for a $t$-biclique-free outerstring graph. 

    The number of instances we produce is $(2t)^{O(k/t)}$, which is $2^{O(\sqrt k\log k)}$.
    For each instance, we compute a biclique of size $t$ in $2^{O(\sqrt k\log k)}$ time. If it does not exist, we compute a minimum feedback vertex set in $O(2^{t(\log t)\log k})=2^{O(\sqrt k\log^2 k)}$ time.   
\end{proof}


\paragraph{Subexponential-time algorithms for (not necessarily sparse) outerstring graphs.}
Apart from Corollary~\ref{cor:biclique-free}, the main tools we use for \textsc{Vertex Cover} and \textsc{Feedback Vertex Set} are branching and kernelization.
For the following two problems, we are not aware of any polynomial kernel
although we can apply branching.
In this case, we can obtain subexponential-time (non-parameterized) algorithms for those problems.

\begin{restatable}{corollary}{mat}\label{cor:matching}
\textsc{Maximum Induced Matching} on outerstring graphs with $n$ vertices can be solved in $2^{O(\sqrt n \log^2 n)}$ time. 
\end{restatable}
\begin{proof}
      Let $G$ be an outerstring graph with $n$ vertices. Let $t=\sqrt n$. 
     We first compute a biclique of size $t$ in $H$ in a brute-force fashion if it exists. This takes $n^t=2^{O(\sqrt n\log n)}$ time. Let $A,B$ be the sets of $V(G)$ which form a biclique of size $t$. Notice that 
     for any induced matching $M$ of $G$,
     at most one edge of $M$ is incident to $A\cup B$. 
     We produce $t^2$ instances for branching on which 
    vertices of $A\cup B$ are involved in a maximum induced matching. 
    For this, we produce 
    two instances $G\setminus A$ and $G\setminus B$,
    and $t^2$ instances $(G\setminus (A\cup B))\cup \{u,v\}$ for all pairs $(u,v)$ with
    $u\in A$ and $v\in B$. 
    If $G$ does not contain a biclique of size $t$, we apply a $2^{O(tw)}n^{O(1)}$-time algorithm for computing a maximum-cardinality induced matching for a $t$-biclique-free outerstring graph.

    The number of instances we produce is $(t^2)^{O(n/t)}$, which is $2^{O(\sqrt n\log n)}$.
    For each instance, we compute a biclique of size $t$ in $2^{O(\sqrt n\log n)}$ time. If it does not exist, we compute a maximum induced matching in $O(2^{t(\log t)\log n})=2^{O(\sqrt n\log^2 n)}$ time. 
\end{proof}

\begin{restatable}{corollary}{listcolor}\label{cor:listcolor}
    \textsc{List 3-Coloring} on outerstring graphs with $n$ vertices can be solved in $2^{O(\sqrt n \log^2 n)}$ time.
\end{restatable}
\begin{proof}
        Let $G$ be an outerstring graph with $n$ vertices, and let $f: V(G)\rightarrow 2^{\{1,2,3\}}$
     be a function. Our goal is to choose one vertex from $f(v)$ for each vertex $v$
     such that for any edge $uv$ of $G$, $f(v)\neq f(u)$. 
     Let $t=\sqrt n$. 
     We first compute a biclique of size $t$ in $H$ in a brute-force fashion if it exists. This takes $n^t=2^{O(\sqrt n\log n)}$ time. Let $A,B$ be the sets of $V(G)$ which form a biclique of size $t$. Notice that 
     for any feasible coloring of $V(G)$,
     either all vertices of $A$ have the same color, or all vertices of $B$ have the same color. This holds since we have only three colors. 
     We produce six instances for branching on which 
    vertices of $A\cup B$ have the same color.
    For this, we produce 
    six instances $(G-A, f_{B\rightarrow i})$ and $(G-B, f_{A\rightarrow i})$
    for $i=1,2,3$, where $f_{A\rightarrow i}$ (and $f_{B\rightarrow i}$) denotes the function that maps
    the vertices $v$ of $V(G)-(A\cup B)$ to $f(v)$ and all vertices $x$ of $B$ (and of $A$) to $f(x)\setminus\{i\}$.
    If $G$ does not contain a biclique of size $t$, we apply a $2^{O(tw)}n^{O(1)}$-time algorithm for computing a list 3-coloring of $G$. 

    The number of instances we produce is $6^{O(n/t)}$, which is $2^{O(\sqrt n)}$.
    For each instance, we compute a biclique of size $t$ in $2^{O(\sqrt n\log n)}$ time. If it does not exist, we compute a list 3-coloring in $O(2^{t(\log t)\log n})=2^{O(\sqrt n\log^2 n)}$ time.   
    In total, the running time of the algorithm is $2^{O(\sqrt n\log^2 n)}$.
\end{proof}

\paragraph{Approximation algorithms for (not necessarily sparse) outerstring graphs.}
We present the first constant-factor approximation algorithm for the \textsc{Cycle Packing} running in quasi-polynomial time.
Our algorithm repeatedly computes a cycle of length at most four and removes it from $G$ until no such cycle exists.
Then the remaining graph is 2-biclique-free, and thus we can find a maximum number of cycles
in the remaining graph in $2^{O(tw\log tw)}$ time. Then we show that the total number of cycles we have found so far
is at least $\textsf{OPT}/4$. 

\begin{restatable}{corollary}{cp}\label{lem:cp}
\textsc{Cycle Packing} on outerstring graphs with $n$ vertices can be solved 
approximately with an approximation factor of 4 in $n^{O(\log\log n)}$ time. 
\end{restatable}
\begin{proof}
    Let $G$ be an outerstring graph with $n$ vertices. 
    As the first step, we repeatedly compute a cycle of length at most four. 
    For a cycle of length at most four, we add it to the solution and remove it from $G$. Then finally we have a graph $G_\textsf{final}$ whose shortest cycle has a length larger than four. 
    Notice that this graph does not contain $K_{2,2}$ as a subgraph, and thus its treewidth is $O(\log n)$. It is known that \textsc{Cycle Packing} can be solved in $2^{O(tw\log tw)}n^{O(1)}$ time.
    As the second step, we find a maximum-cardinality set of vertex-disjoint cycles in $G_{\textsf{final}}$ in $n^{O(\log\log n)}$ time,
    and return all cycles we computed in the first and second steps as output. 

    We can obtain at least $\textsf{OPT}/4$ vertex-disjoint cycles in this way,
    where $\textsf{OPT}$ is the maximum number of vertex-disjoint cycles. 
    Let $\mathcal C_\textsf{OPT}$ be a set of vertex-disjoint cycles with $|\mathcal C_\textsf{OPT}|=\textsf{OPT}$.
    A cycle of $\mathcal C_\textsf{OPT}$ not contained in $G_\textsf{final}$
    intersects at least one cycle we computed in the first step. 
    Since the length of a cycle we computed in the first step is at most four,
    each such cycle intersects at most four cycles of $\mathcal C_\textsf{OPT}$.
    On the other hand, the number of cycles of $\mathcal C_\textsf{OPT}$ contained in $G_\textsf{final}$ is at most the number of cycles we computed in the second step
    since we compute a maximum number of vertex-disjoint cycles of $G_\textsf{final}$
    in the second step.
    Therefore, the number of vertex-disjoint cycles we have  
    is at least $\textsf{OPT}/4$. 
\end{proof}

\bibliographystyle{plain}

\newpage
\bibliography{sample}

\begin{thebibliography}{10}

\bibitem{aboulker2021tree}
Pierre Aboulker, Isolde Adler, Eun~Jung Kim, Ni~Luh~Dewi Sintiari, and Nicolas Trotignon.
\newblock On the tree-width of even-hole-free graphs.
\newblock {\em European Journal of Combinatorics}, 98:103394, 2021.

\bibitem{abrishami2022induced}
Tara Abrishami, Maria Chudnovsky, and Kristina Vu{\v{s}}kovi{\'c}.
\newblock Induced subgraphs and tree decompositions {I}. {E}ven-hole-free graphs of bounded degree.
\newblock {\em Journal of Combinatorial Theory, Series B}, 157:144--175, 2022.

\bibitem{an2023faster}
Shinwoo An, Kyungjin Cho, and Eunjin Oh.
\newblock Faster algorithms for cycle hitting problems on disk graphs.
\newblock In {\em Proceedings of the 18th Algorithms and Data Structures Symposium (WADS 2023)}, pages 29--42, 2023.

\bibitem{an2021feedback}
Shinwoo An and Eunjin Oh.
\newblock Feedback vertex set on geometric intersection graphs.
\newblock In {\em Proceedings of the 32nd International Symposium on Algorithms and Computation (ISAAC 2021)}, 2021.

\bibitem{an2024eth}
Shinwoo An and Eunjin Oh.
\newblock Eth-{T}ight {A}lgorithm for {C}ycle {P}acking on {U}nit {D}isk {G}raphs.
\newblock In {\em 40th International Symposium on Computational Geometry (SoCG 2024)}, 2024.

\bibitem{bhore2022balanced}
Sujoy Bhore, Satyabrata Jana, Supantha Pandit, and Sasanka Roy.
\newblock The balanced connected subgraph problem for geometric intersection graphs.
\newblock {\em Theoretical Computer Science}, 929:69--80, 2022.

\bibitem{biedl2018size}
Therese Biedl, Ahmad Biniaz, and Martin Derka.
\newblock On the size of outer-string representations.
\newblock In {\em Proceedings of the 16th Scandinavian Symposium and Workshops on Algorithm Theory (SWAT 2018)}, pages 10:1--10:14, 2018.

\bibitem{bonamy2023sparse}
Marthe Bonamy, {\'E}douard Bonnet, Hugues D{\'e}pr{\'e}s, Louis Esperet, Colin Geniet, Claire Hilaire, St{\'e}phan Thomass{\'e}, and Alexandra Wesolek.
\newblock Sparse graphs with bounded induced cycle packing number have logarithmic treewidth.
\newblock In {\em Proceedings of the 2023 Annual ACM-SIAM Symposium on Discrete Algorithms (SODA 2023)}, pages 3006--3028. SIAM, 2023.

\bibitem{bonnet2019optimality}
{\'E}douard Bonnet and Pawe{\l} Rz{\k{a}}{\.z}ewski.
\newblock Optimality program in segment and string graphs.
\newblock {\em Algorithmica}, 81:3047--3073, 2019.

\bibitem{bose2022computing}
Prosenjit Bose, Paz Carmi, J~Mark Keil, Anil Maheshwari, Saeed Mehrabi, Debajyoti Mondal, and Michiel Smid.
\newblock Computing maximum independent set on outerstring graphs and their relatives.
\newblock {\em Computational Geometry}, 103:101852, 2022.

\bibitem{cabello2013clique}
Sergio Cabello, Jean Cardinal, and Stefan Langerman.
\newblock The clique problem in ray intersection graphs.
\newblock {\em Discrete \& Computational Geometry}, 50:771--783, 2013.

\bibitem{cardinal2018intersection}
Jean Cardinal, Stefan Felsner, Tillmann Miltzow, Casey Tompkins, and Birgit Vogtenhuber.
\newblock Intersection graphs of rays and grounded segments.
\newblock {\em Journal of Graph Algorithms and Applications}, 22(2):273--295, 2018.

\bibitem{chudnovsky2022induced}
Maria Chudnovsky, Tara Abrishami, Sepehr Hajebi, and Sophie Spirkl.
\newblock Induced subgraphs and tree decompositions {III}. {Three}-path-configurations and logarithmic treewidth.
\newblock {\em Advances in Combinatorics}, 2022:6, 2022.

\bibitem{cygan2015parameterized}
Marek Cygan, Fedor~V Fomin, {\L}ukasz Kowalik, Daniel Lokshtanov, D{\'a}niel Marx, Marcin Pilipczuk, Micha{\l} Pilipczuk, and Saket Saurabh.
\newblock {\em Parameterized Algorithms}, volume~5.
\newblock Springer, 2015.

\bibitem{damaschke1989hamiltonian}
Peter Damaschke.
\newblock The {H}amiltonian circuit problem for circle graphs is {NP}-complete.
\newblock {\em Information Processing Letters}, 32(1):1--2, 1989.

\bibitem{de2020framework}
Mark De~Berg, Hans~L Bodlaender, S{\'a}ndor Kisfaludi-Bak, D{\'a}niel Marx, and Tom~C Van Der~Zanden.
\newblock A framework for exponential-time-hypothesis--tight algorithms and lower bounds in geometric intersection graphs.
\newblock {\em SIAM Journal on Computing}, 49(6):1291--1331, 2020.

\bibitem{de2023clique}
Mark de~Berg, S{\'a}ndor Kisfaludi-Bak, Morteza Monemizadeh, and Leonidas Theocharous.
\newblock Clique-based separators for geometric intersection graphs.
\newblock {\em Algorithmica}, 85(6):1652--1678, 2023.

\bibitem{dvovrak2019treewidth}
Zden{\v{e}}k Dvo{\v{r}}{\'a}k and Sergey Norin.
\newblock Treewidth of graphs with balanced separations.
\newblock {\em Journal of Combinatorial Theory, Series B}, 137:137--144, 2019.

\bibitem{ehrlich1976intersection}
Gideon Ehrlich, Shimon Even, and Robert~Endre Tarjan.
\newblock Intersection graphs of curves in the plane.
\newblock {\em Journal of Combinatorial Theory, Series B}, 21(1):8--20, 1976.

\bibitem{flier2010vertex}
Holger Flier, Mat{\'u}{\v{s}} Mihal{\'a}k, Anita Sch{\"o}bel, Peter Widmayer, and Anna Zych.
\newblock Vertex disjoint paths for dispatching in railways.
\newblock In {\em Proceedings of the 10th Workshop on Algorithmic Approaches for Transportation Modelling, Optimization, and Systems (ATMOS'10)}, 2010.

\bibitem{fox2006bipartite}
Jacob Fox.
\newblock A bipartite analogue of {D}ilworth’s theorem.
\newblock {\em Order}, 23(2-3):197--209, 2006.

\bibitem{fox2010separator}
Jacob Fox and J{\'a}nos Pach.
\newblock A separator theorem for string graphs and its applications.
\newblock {\em Combinatorics, Probability and Computing}, 19(3):371--390, 2010.

\bibitem{fox2012coloring}
Jacob Fox and J{\'a}nos Pach.
\newblock Coloring {$K_k$}-free intersection graphs of geometric objects in the plane.
\newblock {\em European Journal of Combinatorics}, 33:853--866, 2012.

\bibitem{fox2012string}
Jacob Fox and J{\'a}nos Pach.
\newblock String graphs and incomparability graphs.
\newblock {\em Advances in Mathematics}, 230(3):1381--1401, 2012.

\bibitem{friggstad2011approximability}
Zachary Friggstad and Mohammad~R Salavatipour.
\newblock Approximability of packing disjoint cycles.
\newblock {\em Algorithmica}, 60(2):395--400, 2011.

\bibitem{garey1980complexity}
Michael~R Garey, David~S Johnson, Gary~L Miller, and Christos~H Papadimitriou.
\newblock The complexity of coloring circular arcs and chords.
\newblock {\em SIAM Journal on Algebraic Discrete Methods}, 1(2):216--227, 1980.

\bibitem{gavenvciak2018cops}
Tom{\'a}{\v{s}} Gaven{\v{c}}iak, Przemys{\l}aw Gordinowicz, V{\'\i}t Jel{\'\i}nek, Pavel Klav{\'\i}k, and Jan Kratochv{\'\i}l.
\newblock Cops and robbers on intersection graphs.
\newblock {\em European Journal of Combinatorics}, 72:45--69, 2018.

\bibitem{keil1993complexity}
J~Mark Keil.
\newblock The complexity of domination problems in circle graphs.
\newblock {\em Discrete Applied Mathematics}, 42(1):51--63, 1993.

\bibitem{keil2017algorithm}
J~Mark Keil, Joseph~SB Mitchell, Dinabandhu Pradhan, and Martin Vatshelle.
\newblock An algorithm for the maximum weight independent set problem on outerstring graphs.
\newblock {\em Computational Geometry}, 60:19--25, 2017.

\bibitem{kisfaludi2022computing}
S{\'a}ndor Kisfaludi-Bak, Karolina Okrasa, and Pawe{\l} Rz{\k{a}}{\.z}ewski.
\newblock Computing list homomorphisms in geometric intersection graphs.
\newblock In {\em Proceedings of the International Workshop on Graph-Theoretic Concepts in Computer Science (WG 2022)}, pages 313--327. Springer, 2022.

\bibitem{kong2010optimal}
Hui Kong, Qiang Ma, Tan Yan, and Martin~DF Wong.
\newblock An optimal algorithm for finding disjoint rectangles and its application to {PCB} routing.
\newblock In {\em Proceedings of the 47th Design Automation Conference (DAC 2010)}, pages 212--217, 2010.

\bibitem{korhonen2023single}
Tuukka Korhonen.
\newblock A single-exponential time 2-approximation algorithm for treewidth.
\newblock {\em SIAM Journal on Computing}, pages FOCS21--174, 2023.

\bibitem{kratochvil1991string}
Jan Kratochv{\'\i}l.
\newblock String graphs {I}. {The} number of critical nonstring graphs is infinite.
\newblock {\em Journal of Combinatorial Theory, Series B}, 52(1):53--66, 1991.

\bibitem{krivelevich2007approximation}
Michael Krivelevich, Zeev Nutov, Mohammad~R Salavatipour, Jacques Verstraete, and Raphael Yuster.
\newblock Approximation algorithms and hardness results for cycle packing problems.
\newblock {\em ACM Transactions on Algorithms (TALG)}, 3(4):48--es, 2007.

\bibitem{lee2017separators}
James~R Lee.
\newblock Separators in region intersection graphs.
\newblock In {\em Proceedings of the 8th Innovations in Theoretical Computer Science Conference (ITCS 2017)}, 2017.

\bibitem{lokshtanov2022subexponential}
Daniel Lokshtanov, Fahad Panolan, Saket Saurabh, Jie Xue, and Meirav Zehavi.
\newblock Subexponential parameterized algorithms on disk graphs.
\newblock In {\em Proceedings of the 2022 Annual ACM-SIAM Symposium on Discrete Algorithms (SODA 2022)}, pages 2005--2031. SIAM, 2022.

\bibitem{matouvsek2014near}
Ji{\v{r}}{\'\i} Matou{\v{s}}ek.
\newblock Near-optimal separators in string graphs.
\newblock {\em Combinatorics, Probability and Computing}, 23(1):135--139, 2014.

\bibitem{matouvsek2014string}
Ji{\v{r}}{\'\i} Matou{\v{s}}ek.
\newblock String graphs and separators.
\newblock In {\em Geometry, Structure and Randomness in Combinatorics}, pages 61--97. Springer, 2014.

\bibitem{middendorf1993weakly}
Matthias Middendorf and Frank Pfeiffer.
\newblock Weakly transitive orientations, {H}asse diagrams and string graphs.
\newblock {\em Discrete Mathematics}, 111(1-3):393--400, 1993.

\bibitem{okrasa2020subexponential}
Karolina Okrasa and Pawe{\l} Rz{\k{a}}{\.z}ewski.
\newblock Subexponential algorithms for variants of the homomorphism problem in string graphs.
\newblock {\em Journal of Computer and System Sciences}, 109:126--144, 2020.

\bibitem{pach2006comment}
J{\'a}nos Pach and G{\'e}za T{\'o}th.
\newblock Comment on fox news.
\newblock {\em Geombinatorics}, 15:150--154, 2006.

\bibitem{rok2019outerstring}
Alexandre Rok and Bartosz Walczak.
\newblock Outerstring graphs are $\chi$-bounded.
\newblock {\em SIAM Journal on Discrete Mathematics}, 33(4):2181--2199, 2019.

\bibitem{schaefer2004decidability}
Marcus Schaefer and Daniel {\v{S}}tefankovi{\v{c}}.
\newblock Decidability of string graphs.
\newblock {\em Journal of Computer and System Sciences}, 2(68):319--334, 2004.

\bibitem{Sinden1966TopologyOT}
Frank~W. Sinden.
\newblock Topology of thin film {RC} circuits.
\newblock {\em Bell System Technical Journal}, 45:1639--1662, 1966.

\end{thebibliography}

\newpage
\appendix

\end{document}